\documentclass[reqno,10pt]{article}
\usepackage{graphicx}
\usepackage{amsmath}
\usepackage{amssymb}
\usepackage{amsthm}
\usepackage{enumitem}
\usepackage{bbm}
\usepackage[l2tabu,orthodox]{nag}
\usepackage[all,error]{onlyamsmath}
\usepackage[strict]{csquotes}
\usepackage{url}
\newtheorem{theorem}{Theorem}

\theoremstyle{definition}

\numberwithin{equation}{section}
\numberwithin{theorem}{section}

\newenvironment{OMabstract}{\noindent\textbf{Abstract.} }{\medskip}
\newenvironment{OMsubjclass}{\noindent\textbf{Mathematics Subject Classification (2020):} }{\medskip}
\newenvironment{OMkeywords}{\noindent\textbf{Keywords:}  }{\medskip}

\begin{document}

\author{Vyacheslav Pivovarchik} 
\title{ Recovering the shape of a quantum tree by two spectra}
\maketitle


\begin{OMabstract}
We show  how to find the shape of an equilateral tree using  the spectra of the Neumann and the Dirichlet problems generated by the Sturm-Liouville equation.   
In case of snowflake trees the spectra of the Neumann and Dirichlet problems uniquely determine the shape of the tree.      
\end{OMabstract}

\begin{OMkeywords}
      Sturm-Liouville equation, eigenvalue, equilateral tree, star graph, Dirichlet boundary condition, Neumann boundary condition.
\end{OMkeywords}

\begin{OMsubjclass}
    34B45, 34B240, 34L20
\end{OMsubjclass}


\section{Introduction}  
The problem of recovering the shape of a combinatorial graph using the eigenvalues of its adjacensy matrix is  described in \cite{CDZ} where several examples of cospectral graphs are shown.

In quantum graph theory the problem of recovering the shape of a graph was stated  in \cite{vB} and \cite{GS}. It was shown in \cite{GS} that if the lengths of the edges are non-commeasurate then the spectrum of the spectral Sturm-Liouville problem on a graph with standard conditions at its vertices uniquely determines the shape of this graph. In \cite{vB}, it was shown that in case of commensurate lengths of the edges there exist cospectral quantum graphs. 
In \cite{KN} it was shown that the spectrum of the Neumann problem with zero potential  on $P_2$ uniquely determines the shape of the graph.  
In \cite{CP} it was shown that if the graph is simple connected equilateral with the number of vertices less or equal 5 and the potentials on the edges are real $L_2$ functions then the spectrum of the Sturm-Liouville problem with standard conditions at the vertices uniquely determines the shape of the graph. For trees the minimal number of vertices in a cospectral pair is 9 (see \cite{Pist} and \cite{CP1}).  If the number of vertices doesn't exceed 8 then to find the shape of a tree we need just to find in \cite{CP} the characteristic polynomial corresponding to the given spectrum. 
 
In present paper we show how to find the shape of a tree using the spectra of the Dirichlet and the Neumann problems. This method works even in case of large number of vertices. If the solution is not unique we can find all the solutions. 

In Section 2 we describe the Neumann spectral problem, i.e. the Sturm-Liouville problem with standard conditions (continuity + Kirchhoff at the interior vertices and Neumann at the pendant vertices). Also we describe the Dirichlet problem where we impose the Dirichlet condition at the root (an arbitrary chosen vertex) keeping standard conditions at all the other vertices. 
Also  we expose known results  which we use in the sequel.

In Section 3 we prove a theorem where the fraction of the characteristic polynomial of the normalized Laplacian  of the corresponding  combinatorial tree and the modified characteristic polynomial of  its certain subgraph (a tree or a forest, obtained by deleting the root and the incident edges)  is presented as a branched continuous fraction. In case of a snowflake tree this presentation is unique.

In Section 4 using the result of Section 3 we show the procedure of recovering the shape of a tree using asymptotics of the spectra of the Neumann and Dirichlet problems  and the degree of the root.

\section{Statement of the problem and auxiliary results}

Let $T$ be an equilateral  tree with $p$ vertices and $q=p-1$ edges each of the length $l$. We choose an arbitrary  vertex $v_0$ as the root and direct all  the edges  away from the root.  
Let us describe the {\it Neumann} spectral  problem on this tree. We consider  the Sturm-Liouville equations on the edges 
\begin{equation}
\label{2.1}
-y_j^{\prime\prime}+q_j(x)y_j=\lambda y_j, \ \  j=1,2,..., g 
\end{equation} 
where $q_j\in L_2(0,l)$ are real.

For each edge $e_j$ incident with a pendant vertex which is not the root 
we impose the Neumann condition 
\begin{equation}
\label{2.2}
y_j'(l)=0. 
\end{equation}
At  each interior vertex which is not the root we impose the continuity conditions   
\begin{equation}
\label{2.3}
y_j(l)=y_k(0)
\end{equation}
for the incoming into $v_i$ edge $e_j$ and for all $e_k$ outgoing from $v_i$ and the Kirchhoff's conditions
\begin{equation}
\label{2.4}
y'_j(l)=\mathop{\sum}\limits_k y_k'(0)
\end{equation}
where the sum is taken over all edges $e_k$ outgoing from $v_i$. 

If the root is an interior vertex then the conditions at $v_0$ are
\begin{equation}
\label{2.5}
y_i(0)=y_j(0)
\end{equation}
for all indices $i$ and $j$ of the edges incident with the root and
\begin{equation}
\label{2.6}
\mathop{\sum}\limits_i y_i'(0)=0.
\end{equation}
If the root is pendant and  its incident edge is $e_1$ then
\begin{equation}
\label{2.7}
y_1'(0)=0.
\end{equation}
The above conditions (continuity +Kirchhoff''s or Neumann) we call standard.


{\bf Standing assumption}
For all edges the potentials $q_j$ are real-valued functions of the space $L^2(0,\ell)$.

In the sequel, if the potentials are the same on all the  edges we omit the index in $q_j$ and $y_j$. The following theorem adopted for trees can be found as Theorem 5.2 in \cite{CP} but it originates from \cite{Bel85}.


{\bf Theorem  2.1}  
{\it
 Let $T$ be a tree with $p\geq 2$.  Assume that all edges have the same length $l$ and the same potentials symmetric with respect to the midpoints of the edges ($q(l-x)=q(x)$).  Then the spectrum of problem (\ref{2.1})--(\ref{2.6})  or (\ref{2.1})--(\ref{2.4}), (\ref{2.7}) coincides with the set of zeros of the function  
\begin{equation}
\label{2.8} 
\phi_N(\lambda)=s(\sqrt{\lambda}, l) \tilde{\psi}(c (\sqrt{\lambda}, l))
\end{equation} 
where $\tilde{\psi}(z)=(1-z^2)^{-1}\psi(z)$,
\[\psi(z)=det(-zD+A).\] 
Here $A$ is the adjacency matrix of $T$, 
\[D=diag(d(v_0), d(v_1),...,d(v_{p-1}),\]  
$d(v_i)$ is the degree of the vertex $v_i$, $s(\sqrt{\lambda},x)$ and $c(\sqrt{\lambda},x)$ are the solutions of the Sturm-Liouville equation on the edges satisfying the  conditions $s(\sqrt{\lambda},0)=s'(\sqrt{\lambda},0)-1=0$ and $c(\sqrt{\lambda},0)-1=c'(\sqrt{\lambda},0)=0$.}

Now we consider the Dirichlet problem. We impose the Dirichlet condition at $v_0$: 
\begin{equation}
\label{2.9}
y_i(0)=0
\end{equation}
for all edges incident with $v_0$, and consider the Dirichlet problem which consists of equations (\ref{2.1})--(\ref{2.4}) and (\ref{2.9}).   

Then 
 we can consider  $T$ as a union of $d(v_0)$  subtrees   $T_1$, $T_2$, ..., $T _{d(v_0)}$ which have common vertex $v_0$ and spectral problems on them meaning that  the Dirichlet conditions are imposed at $v_0$ while at the rest of  vertices we keep the standard  conditions. Thus, we have $d(v_0)$ problems on the subtrees.  
 
 Denote by $\hat{T}_i$ the tree obtained by removing the pendant vertex with the Dirichlet boundary conditions (the root) and the edge  incident with it in $T_i$. 
 Let $\hat{A}_i$  be the adjacency matrix of ${\hat T}_i$, let
$\hat{D}_{T,i}=diag \{d(v_{i,1}), d(v_{i,2}), ..., d(v_{i,p_i-1})\}$, where 
$d(v_{i,j})$ is the degree of the vertex $v_{i,j}$ in $T_i$ (we underline that in $T_i$, not in $\hat{T}_i$ !) 
and $p_i$ is the number of vertices $\{v_0, v_{i,1}, ..., v_{i,p_i-1}\}$ in $T_i$.

We consider the polynomials  defined by 
\begin{equation}
\label{2.10}
\hat{\psi}_{i}(z):=det(z\hat{D}_{T,i}-\hat{A_i}).
\end{equation}
Theorem 6.4.2 of  \cite{MP2}  adapted to the case a tree with the Dirichlet condition at one of the vertices is as follows

{\bf Theorem 2.2} {\it Let $T_i$ be a tree with at least two edges rooted at a pendant vertex $v_0$. Let the Dirichlet condition be imposed at the root and the standard conditions at all other vertices. Assume that all edges have the same length $l$ and the same potentials symmetric with respect to the midpoints of the edges ($q(l-x)=q(x)$). Then the spectrum of problem (\ref{2.1})--(\ref{2.4}), (\ref{2.9})  coincides with the set of zeros of  the characteristic function
\begin{equation}
\label{2.11}
\phi_{D,i}(\lambda)=\hat{\psi}_i(c(\sqrt{\lambda},l)).
\end{equation}
 }


It is clear that 
\begin{equation}
\label{2.12}
\phi_D(\lambda)=\prod_{i=1}^{d(v_0)}\phi_{D,i}(\lambda)=\prod_{i=1}^{d(v_0)}det(c(\sqrt{\lambda},l)\hat{D}_{T,i}-\hat{A_i})
\end{equation} 
is the characteristic function of the Dirichlet problem  (\ref{2.1})--(\ref{2.4}), (\ref{2.9}) on the initial tree $T$.

\vspace{3mm}

Denote by  
\[
\hat{\psi}(z):=\prod_{i=1}^{d(v_0)}\hat{\psi}_{i}(z). 
\]

It is clear that
\[
\hat{\psi}(z)=det(-z\hat{D}+\hat{A})
\]
where   $\hat{A}$  be the adjacency matrix of the  forest  $\hat{T}$ (set of subtrees $\hat{T}_1$, $\hat{T}_2$, ..., $\hat{T} _{d(v_0)}$ obtained by deleting $v_0$ and the incident edges from $T$ (we call {\it principal subforest (subtree if $d(v_0)=1$})),  
$\hat{D}_{T}=diag \{d(v_{1}), d(v_{2}), ..., d(v_{p-1})\}$, where 
$d(v_{j})$ is the degree of the vertex $v_{j}$ in $T$ and $p$ is the number of vertices $\{v_0, v_{1}, ..., v_{p-1}\}$ in $T$ (see Fig. 1).


\section{Main results}

First of all we notice that $-z\hat{D}+\hat{A}$ is the principal submatrix of matrix $-zD+A$ obtained by deleting the row and column corresponding to $v_0$.

\newpage

\begin{figure}
 \begin{center}
   \includegraphics[scale= 0.7 ] {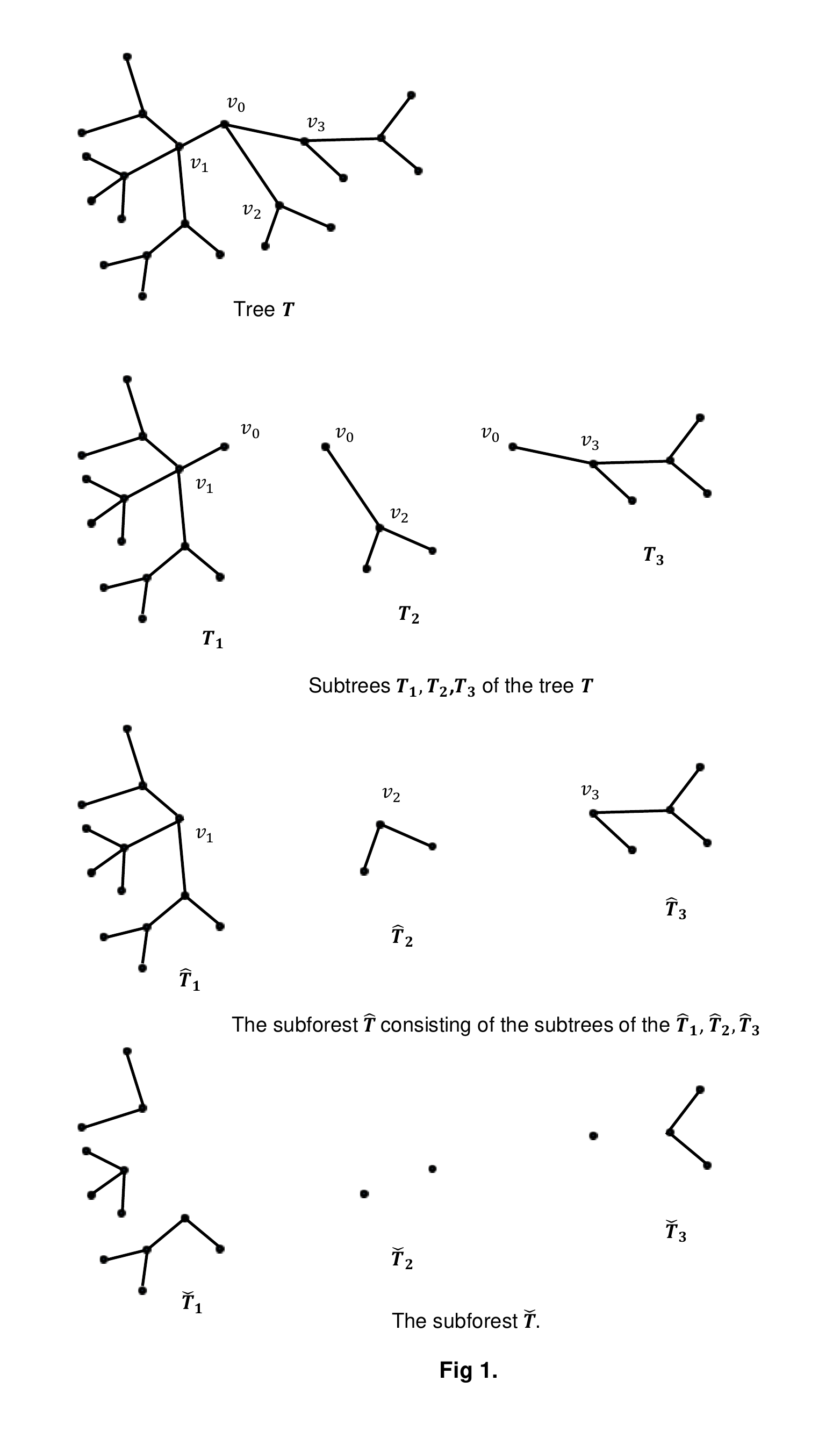}
  \end{center}
\end{figure}


\begin{theorem}
Let $T$ be an equilateral tree. Then the fraction $\frac{\psi(z)}{\hat{\psi}(z)}$ can be  expanded in branched continuous fraction. The coefficients before $+z$ or $-z$ correspond to the degrees of  the vertices. The beginning fragment 
\[
-m_0z+
\sum_{k=1}^{m_0} \frac{1}{
m_kz-...}
\]
of the expansion means that the vertex $v_0$ is connected by edges with  $m_0$ vertices, say $v_1$, $v_2$, ..., $v_{m_0}$.

A fragment
\[
...\pm\sum_{i=1}^{r}\frac{1}{-m_iz+\sum_{k=1}^{m_i-1}\frac{1}{+m_{i,k}z+...}}
\]
means that there are $r$ vertices each have one incoming edge and $m_i-1$ ($i=1,...,r$)
 outgoing edges.

A fragment 
\[
... \pm\frac{m}{z}
\]
at an end of a branch of the continuous fraction means $m$ edges ending with  pendant vertices.
\end{theorem}

\begin{proof} 

 Consider the matrix $-zD+A$ meaning that the first row corresponds to $v_0$ and the next $d(v_0)$  rows to the vertices adjacent with $v_0$. The first row expansion of the determinant  of this matrix $-zD+A$ gives
\begin{equation}
\label{3.2}
det(-zD+A)=-d(v_0)z det(-z\hat{D}+\hat{A}) - \sum_{k=1}^{d(v_0)}det(-z\check{D}_k+\check{A}_k)
\end{equation}
where the principal submatrix $(-z\hat{D}+\hat{A})$ is obtained from $(-zD+A)$ by deleting the first row and the first column while $(-z\check{D}_k+\check{A}_k)$ is the obtained from $(-z\hat{D}+\hat{A})$ by deleting its $k$th row and its $k$th column. The corresponding subtrees can be seen at Fig. 1.

Deviding both parts of (\ref{3.2}) by $det(-z\hat{D}+\hat{A})$ we continue expanding into branched fraction and  obtain 

\begin{equation}
\label{**}
\frac{\psi(z)}{\hat{\psi}(z)}=\frac{det(-zD+A)}{det(-z\hat{D}+\hat{A})}=
-d(v_0)z-\frac{\sum_{i=1}^{d(v_0)}det(-z\check{D}_i+\check{A}_i)}{det(-z\hat{D}+\hat{A})}=
\end{equation}
\[
-d(v_0)z+
\sum_{i=1}^{d(v_0)} \frac{1}{
d(v_i)z-
\sum_{i=1}^{d(v_i) - 1}
\frac{
\hat{\psi}_{i}(z)
}  
{ \check{\psi}_{i}(z)}
}
\]
Here $\check{\psi}_{i}(z)$ is the modified characteristic polynomial of the subtree $\check{T}_i$. To finish the proof we need to continue this procedure.
\end{proof}

{\bf Example} Let $\psi(z)=-108z^6+258z^4-203z^2+52$ and $\hat{\psi}(z)=36z^5-58z^3+23z$. Then 
\[
\frac{\psi(z)}{\hat{\psi}(z)}=-3z+\frac{84z^4-133z^2+12}{36z^5-58z^3+23z}=-3z+\frac{1}{z}+\frac{1}{z}+\frac{12z^4-17z^2+6}{36z^5-58z^3+23z}=
\]
\[
-3z+\frac{1}{z}+\frac{1}{z}+\frac{1}{3z-\frac{7z^3-5z}{12z^4-17z^2+6}}=
-3z+\frac{1}{z}+\frac{1}{z}+\frac{1}{3z-\frac{z}{-3z^2+2}-\frac{z}{-4z^2+3}}=
\]
\begin{equation}
\label{3.3}
-3z+\frac{1}{z}+\frac{1}{z}+\frac{1}{3z-\frac{1}{-3z+\frac{1}{z}+\frac{1}{z}}-\frac{1}{-4z+\frac{1}{z}+\frac{1}{z}+\frac{1}{z}}}.
\end{equation}
Thus, this branched continued fraction corresponds to the tree of Fig. 2. 

\begin{figure}
 \begin{center}
   \includegraphics[scale= 0.9 ] {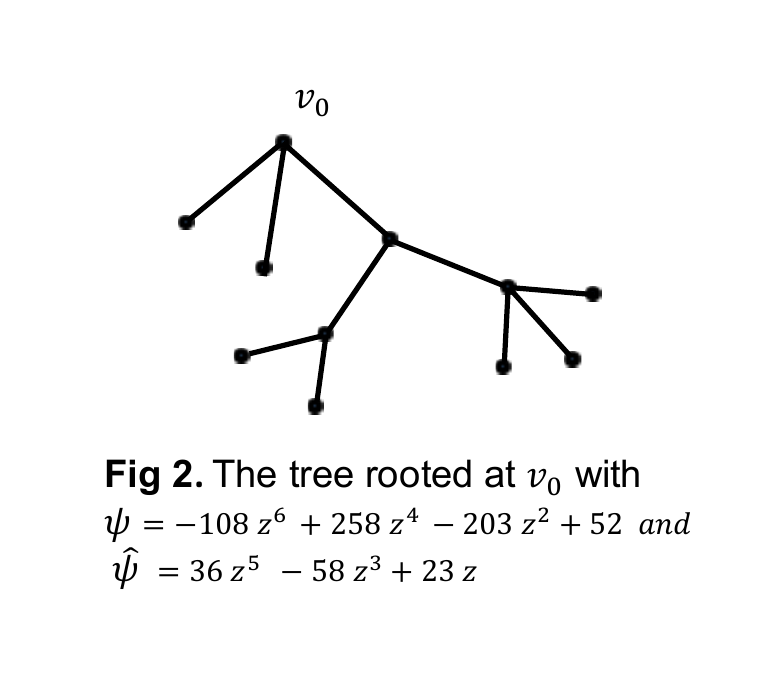}
  \end{center}
\end{figure}

This tree   is unique corresponding to $\psi(z)=-108z^6+258z^4-203z^2+52$ and $\hat{\psi}(z)=36z^5-58z^3+23$. To prove it we notice that 
\[
\lim\limits_{z\to \infty}=-\frac{1}{z}\frac{\psi(z)}{\hat{\psi}(z)}=3
\]
what means that the degree of the root is 3. Now 
\[
\lim\limits_{z\to \infty}z\left(\frac{\psi(z)}{\hat{\psi}(z)}+3z\right)=\frac{7}{3}.
\]
Therefore,
\[
\frac{7}{3}=\frac{1}{d_1}+\frac{1}{d_2}+\frac{1}{d_3}
\]
 where $d_1$, $d_2$ and $d_3$ are the degrees of the vertices adjacent with the root. This equation has the only (up to permutations) solution in natural numbers $\{1,1, 3\}$. That means that the expansion
\[\frac{\psi(z)}{\hat{\psi}(z)}=
-3z+\frac{1}{z}+\frac{1}{z}+\frac{1}{3z-\frac{7z^3-5z}{12z^4-17z^2+6}}
\]
is unique. Now we need to  expand $\frac{7z^3-5z}{12z^4-17z^2+6}$ in the form of two summands, Since the equation $\frac{1}{d_3}+\frac{1}{d_4}=\frac{7}{12}$ has the only (up to permutation) solution $\{3, 4\}$. The only way to expand the fraction  $\frac{7z^3-5z}{12z^4-17z^2+6}$ is shown in (\ref{3.3}).

\vspace{3mm}
By snowflake graph we mean a tree with the distance between the root and any pendant vertex $\leq 2$ (see an example at Fig. 3). 

\begin{figure}
 \begin{center}
   \includegraphics[scale= 0.7 ] {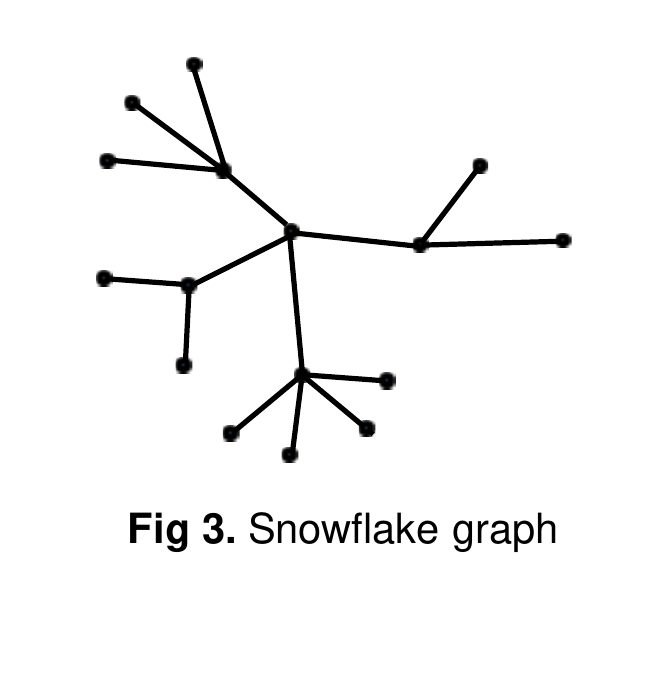}
  \end{center}
\end{figure}

{\bf Theorem 3.2} Let $T$ be a snowflake graph rooted at the central vertex. The corresponding  two functions  $\psi(z)$ and $\hat{\psi}(z)$ uniquely determine the shape of the graph. 

\begin{proof}
 In case of snowflake graph equation (\ref{**}) looks as follows
\begin{equation}
\label{***}
\frac{\psi(z)}{\hat{\psi}(z)}=
-d(v_0)z+
\sum_{k=1}^{d(v_0)} \frac{1}{
d(v_k)z-\frac{d(v_k)-1}{z}}.
\end{equation}
Thus, we can find the degree of the central vertex: 
\begin{equation}
\label{3.5}
\lim\limits_{z\to \infty}\left(-\frac{1}{z}\frac{\psi(z)}{\hat{\psi}(z)}\right)=d(v_0)
\end{equation}
After it using equation
\begin{equation}
\label{3.6}
\frac{\psi(z)}{\hat{\psi}(z)}+d(v_0)z=
z\frac{\sum_{k=1}^{d(v_0)}\prod_{s=1, s\not=k}^{d(v_0)}
(d(v_s)z^2-d(v_s)+1)}
{\prod_{k=1}^{d(v_0)}
(d(v_k)z^2-d(v_k)+1)}.
\end{equation}
we can find the denominator of the right-hand side of (\ref{3.6}), i.e. the polynomial 
$\prod_{k=1}^{d(v_0)}
(d(v_k)z^2-d(v_k)+1)$. Expantion of this polynomial as a  product is unique.
\end{proof}

\section{ Recovering the shape of a quantum graph tree by two spectra}

Now we are ready to recover the shape of a tree.  Using the asymptotics of the spectrum of the Neumann problem we can find  the function $\psi(\lambda)$ (up to a constant factor). Let us show it.  

The eigenvalues of problem (\ref{2.1})--(\ref{2.6}) or (\ref{2.1})--(\ref{2.4}), (\ref{2.7}) can be presented as the union of subsequences $\{\lambda_k\}_{k=1}^{\infty}=\mathop{\cup}\limits_{i=1}^{2g-1}\{\lambda_k^{(i)}\}_{k=1}^{\infty}$
with the following asymptotics
\begin{equation}
\label{4.1}
\sqrt{\lambda_k^{(i)}}\mathop{=}\limits_{k\to\infty}\frac{2\pi (k-1)}{l}+\frac{1}{l}\arccos \alpha_{i}+O\left(\frac{1}{k}\right) \ \ {\rm for}  \ \  i=2,3,..., p-1, \ \ k\in \mathbb{N}
\end{equation} 
\[
\sqrt{\lambda_k^{(i)}}\mathop{=}\limits_{k\to\infty}\frac{2\pi k}{l}-\frac{1}{l}\arccos \alpha_{i}+O\left(\frac{1}{k}\right) \ \ {\rm for}  \ \  i=p, p+1,..., 2p-3, \ \ k\in \mathbb{N}
\]
\begin{equation}
\label{4.2} 
\sqrt{\lambda_k^{(1)}}\mathop{=}\limits_{k\to\infty}\frac{\pi (k-1)}{l}+O\left(\frac{1}{k}\right) \ \  k\in\mathbb{N}.
\end{equation}
where $\alpha_1=1\leq \alpha_2\leq \alpha_3\leq ... \leq \alpha_{p-1}\leq \alpha_p=1$ are the zeros of $\psi(z)$.

Now by Theorem 2.2 and equation (\ref{2.12}) we see that the spectrum \\
 $\{\nu_k\}_{k=1}^{\infty}=\mathop{\cup}\limits_{i=1}^{d(v_0)}\mathop{\cup}\limits_{j=1}^{p_i-1}\{\nu_{k,j}^{(i)}\}_{k=1}^{\infty}$ 
of the Dirichlet problem (\ref{2.1})--(\ref{2.4}), (\ref{2.9}) on $T$ is the union of the spectra of the Dirichlet problems on the subtrees $T_i$ ($i=1,...,d(v_0)$). According to Theorem 2.2, the spectrum of the Dirichlet problem on $T_i$  consists of the subsequences 
\begin{equation}
\label{4.3}
\sqrt{\nu_{k,j}^{(i)}}\mathop{=}\limits_{k\to\infty}\frac{2\pi (k-1)}{l}+\frac{1}{l}\arccos \beta_{i,j}+O\left(\frac{1}{k}\right) \ \ {\rm for} 
\end{equation}
 \[ 
 j=1,2,..., p_i-1,  \ \ k\in \mathbb{N}
\]
\begin{equation}
\label{4.4}
\sqrt{\nu_{k,j}^{(i)}}\mathop{=}\limits_{k\to\infty}\frac{2\pi k}{l}-\frac{1}{l}\arccos \beta_{i,j}+O\left(\frac{1}{k}\right) \ \ {\rm for} 
\end{equation}
 \[ 
 j=p_i, p_i+1, ..., 2p_i-3, \ \ k\in \mathbb{N}
\]
where $\beta_{i,j}$ are the zeros of $\hat{\psi}_{i}(z)$.

{\bf Theorem 4.1} {\it Let $\{\lambda_k\}_{k=1}^{\infty}$ be the spectrum of the Neumann problem (\ref{2.1})--(\ref{2.6}) or (\ref{2.1})--(\ref{2.4}), (\ref{2.7}) and $\{\nu_k\}_{k=1}^{\infty}=\mathop{\cup}\limits_{i=1}^{d(v_0)}\mathop{\cup}\limits_{j=1}^{p_i-1}\{\nu_{k,j}^{(i)}\}_{k=1}^{\infty}$ 
be the spectrum of the Dirichlet problem (\ref{2.1})--(\ref{2.4}), (\ref{2.9}) where the Dirichlet condition is imposed at a vertex $v_0$ of degree $d(v_0)$.
Let $\{\alpha_k\}_{k=1}^p$ be the constants in (\ref{4.1}) and $\{\beta_{i,j}\}_{i=1, j=1}^{d(v_0),p_i}$ the constants in (\ref{4.3}).

Then 
\begin{equation}
\label{4.4}
\frac{\psi(z)}{\hat{\psi}(z)}=d(v_0)\frac{\prod_{i=1}^p(-z+\alpha_i)}{\prod_{i=1}^{d(v_0)}\prod_{j=1}^{p_i-1}(-z+\beta_{i,j})}.
\end{equation}
}

\begin{proof}  By Theorem 2.1 we know that  $\{\alpha_k\}_{k=1}^p$ is the set of zeros of $\psi(z)$ and by Theorem 2.2  that $\{\beta_{i,j}\}_{i=1, j=1}^{d(v_0),p_i}$ is the set of zeros of $\hat{\psi}(z)$. Thus,  we conclude that
\[
\frac{\psi(z)}{\hat{\psi}(z)}=C\frac{\prod_{i=1}^p(-z+\alpha_i)}{\prod_{i=1}^{d(v_0)}\prod_{j=1}^{p_i-1}(-z+\beta_{i,j})},
\]
where $C$ is a nonzero constant. By (\ref{3.5}) we obtain $C=d(v_0)$. 
\end{proof}

Suppose the degree of the root $d(v_0)$ is given. Then we can find  $\{\alpha_k\}_{k=1}^p$ and $\{\beta_{i,j}\}_{i=1, j=1}^{d(v_0),p_i}$ from asymptotics (\ref{4.1}) and (\ref{4.3}) and then expanding (\ref{4.4}) into branched continued fraction find the shape of the tree. We do not state that this tree is always unique.




\begin{thebibliography}{99}



\bibitem{vB} 
{J.\ von} Below.
\newblock Can one hear the shape of a network?
\newblock In F.\ {Ali Mehmeti}, {J.\ von} Below, and S.\ Nicaise, editors, {\em
 Partial Differential Equations on Multistructures (Proc.\ Luminy 1999)},
 volume 219 of {\em Lect.\ Notes Pure Appl.\ Math.}, pages 19--36, New York,
 2001. Marcel Dekker.
%
\bibitem{Bel85}
{J.\ von} Below.
\newblock A characteristic equation associated with an eigenvalue problem on
 $c^2$-networks.
\newblock {\em Lin.\ Algebra Appl.} (1985) 71:309--325.





\bibitem{CDZ} D.M.\ Cvetkovi\'c, M.\ Doob, and H.\ Sachs.
\newblock {\em {Spectra of Graphs -- Theory and Applications}}.
\newblock Pure Appl.\ Math. Academic Press, New York, 1979.





\bibitem{CP} A. Chernyshenko, V. Pivovarchik. Recovering the shape of a quantum graph. Integr. Equ. Oper. Theory  (2020) 92:23. 

\bibitem{CP1} A. Chernyshenko, V. Pivovarchik. Cospectral quantum graphs. 2022 arXiv:2112.14235


\bibitem{GS} B. Gutkin, U. Smilansky, Can one hear the shape of a graph? J. Phys. A Math. Gen. (2001), 34:6061--6068. 
   

\bibitem{KN} P. Kurasov, S. Naboko. \textit{Rayleigh estimates for differential operators on graphs}. J. Spectr. Theory, \textbf{4} (2014), no. 2, 211--219. DOI 10.4171/JST.




\bibitem{MP2} M. M\"oller, V. Pivovarchik, \textit{Direct and inverse finite-dimensional spectral problems on graphs}.
Oper. Theory: Adv., Appl., \textbf{283}. Birkhäuser/Springer,  2020.  ISBN: 978-3-030-60483-7; 978-3-030-60484-4 https://www.springer.com/gp/book/9783030604837.
 
\bibitem{MuP} D. Mugnolo, V. Pivovarchik.  Distinguishing co-spectral quantum graphs by scattering, J. Phys. A: Math. Theor., Vol. 56, issue 9, (2023) DOI: 10.1088/1751-8121/acbb44.  




\bibitem{Pist} M.-E.~Pistol. Generating isospectral but not isomorphic quantum graphs. arXiv: 2104.12885.


\end{thebibliography}
\end{document}